\documentclass[lettersize, journal, 10pt]{IEEEtran}
\usepackage[T1]{fontenc}

\usepackage{booktabs}
\usepackage{cite}

\usepackage[dvipsnames*,svgnames]{xcolor}

\ifCLASSINFOpdf
\else
\fi
\usepackage{amsmath}
\usepackage{amsfonts}
\usepackage{amssymb}
\usepackage{amsthm}
\usepackage{bm}
\usepackage{graphicx}
\usepackage{subfigure}
\usepackage{stfloats}	
\usepackage[dvipsnames*,svgnames]{xcolor}
\usepackage{comment}
\usepackage{ulem}		
\usepackage{tikz-cd}
\usepackage{stmaryrd}
\usepackage{multirow}
\usepackage{makecell}
\usepackage{gensymb}

\usepackage{algorithm}
\usepackage{algorithmicx}
\usepackage{algpseudocode}

\makeatletter

\newcommand{\Rmnum}[1]{\expandafter\@slowromancap\romannumeral #1@}
\makeatother

\newcommand{\RNum}[1]{\uppercase\expandafter{\romannumeral #1\relax}}

\newcommand\redsout{\bgroup\markoverwith{\textcolor{red}{\rule[0.5ex]{1pt}{1.5pt}}}\ULon}

\makeatletter
\newcommand{\removelatexerror}{\let\@latex@error\@gobble}
\makeatother

\newcommand{\Input}{\textbf{Input: }} 
\newcommand{\Output}{\textbf{Output: }}
\newcommand{\Break}{\textbf{break}}
\newcommand{\IF}{\textbf{if }}
\newcommand{\THEN}{\textbf{then }}
\newcommand{\ELSEIF}{\textbf{else if }}
\newcommand{\ENDIF}{\textbf{end if}}

\makeatletter
\renewcommand{\ALG@beginalgorithmic}{\footnotesize}
\makeatother

\theoremstyle{definition}
\newtheorem{proposition}{Proposition}
\theoremstyle{corollary}

\theoremstyle{remark}

\newcommand{\proposname}{Proposition }
\newcommand{\algorithmname}{Algorithm }

\AtBeginDocument{
 	\addtolength{\abovedisplayskip}{-1.0ex}
    \addtolength{\abovedisplayshortskip}{-1.0ex}
	\addtolength{\belowdisplayskip}{-1.0ex}
	\addtolength{\belowdisplayshortskip}{-1.0ex}
}

\begin{document}
%
\title{{Hybrid Beamforming Design for Millimeter Wave Multiuser MIMO Systems with Dynamic Subarrays}}
%
%
%

\author{
		Gengshan Wang, 
		Zhijia Yang, 
		and Tierui Gong 
		\thanks{G. Wang and Z. Yang are with the State Key Laboratory of Robotics, Shenyang Institute of Automation, Chinese Academy of Sciences, Shenyang 110016, China, also with the Key Laboratory of Networked Control System, Shenyang Institute of Automation, Chinese Academy of Sciences, Shenyang 110016, China, and also with the Institutes for Robotics and Intelligent Manufacturing, Chinese Academy of Sciences, Shenyang 110169, China (email: wanggengshan@sia.cn, yang@sia.cn). G. Wang is also with the University of Chinese Academy of Sciences, Beijing 100049, China. T. Gong is with Singapore University of Technology and Design, Singapore (email: tierui\_gong@sutd.edu.sg). }
		\vspace{-0.5cm}
}

\maketitle

\begin{abstract}
In this letter, we investigate the millimeter wave (mmWave) downlink multiuser multiple-input multiple-output (MU-MIMO) system, adopting the dynamic subarray architecture at the base station and considering the multi-stream communication for each user. 
Aiming at maximizing the system spectral efficiency, we propose a novel hybrid beamforming design. First, assuming no inter-user interference (IUI), we easily get the optimal fully-digital beamformers and combiners using the singular value decomposition of each user channel and the waterfilling algorithm. Then, based on the obtained fully-digital beamformers, we propose a Kuhn-Munkres algorithm-assisted dynamic hybrid beamforming design, which guarantees that each radio-frequency chain is connected to at least one antenna. 
Finally, we propose to further project each obtained digital beamformer onto the null space of all the other equivalent user channels to cancel the IUI. Numerical results verify the superiority of our proposed hybrid beamforming design. 
\end{abstract}


%
\IEEEpeerreviewmaketitle

\vspace{-0.3cm}
\section{Introduction}
Millimeter wave (mmWave) has been deemed as a key enabling technology for future wireless communications, as they can provide abundant underutilized spectrum resources to achieve higher data rates \cite{Heath_Overview_mmWave}. The short wavelength at mmWave frequency enables a large number of antennas to be packed in a small dimension, which facilitates the utilization of massive multiple-input multiple-output (MIMO) technique in mmWave systems to realize large-scale spatial multiplexing and perform highly directional beamforming \cite{Heath_Overview_mmWave, MassiveMIMOBS_First, Tse_FoWC}. 

Given the high cost and power consumptions of mmWave radio-frequency (RF) chains, the conventional fully-digital beamforming that requires one dedicated RF chain for each antenna is infeasible for mmWave massive MIMO systems. To address the mmWave hardware constraints while achieving good performance, the hybrid beamforming, using a reduced number of RF chains, has been proposed \cite{Heath_Overview_mmWave}.  
According to the connections between RF chains and antennas, the hybrid beamforming structures can be categorized as the fully-connected structure (FCS) and partially-connected structure (PCS). In the FCS, each RF chain connects to all antennas via phase shifters, making this structure enjoy full beamforming gains \cite{ JunZhang_MIMO_OFDM_HB, HBF_MUMIMO_BUPT}. 
The FCS has higher hardware efficiency than fully-digital beamforming but requires large numbers of phase shifters.  
To further reduce the hardware complexity, the PCS has been proposed. Most existing works focus on the fixed subarrays (FSs) that connect each RF chain to a fixed disjoint subset of antennas \cite{Heath_Overview_mmWave}. 
The FSs greatly reduce the number of phase shifters, while some beamforming gains are sacrificed.

To improve the performance of the PCS, the dynamic subarrays (DSs) where each RF chain is connected to a nonempty disjoint subset of dynamically partitioned antennas have been studied \cite{Heath_SubA_HB, mmWaveWB_Dynamic_DUT, Dynamic_LowCom, Dyanmic_1Ant}. In \cite{Heath_SubA_HB}, the authors first proposed the DS architecture and presented the channel covariance matrix based hybrid beamforming. The authors in \cite{mmWaveWB_Dynamic_DUT} exploited the penalty dual decomposition (PDD) method to get the dynamic hybrid beamforming with low-resolution phase shifters. However, these two works cannot guarantee that each RF chain is connected to at least one antenna, which may result in the problem that the number of used RF chains is smaller than that of data streams. 
The work in \cite{Dynamic_LowCom} proposed a low complexity greedy hybrid beamforming for DSs. 
However, all the above works only consider point-to-point communications. In \cite{Dyanmic_1Ant}, the authors considered the multiuser (MU) scenario and proposed the subchannel difference based dynamic hybrid beamforming. However, this paper only considers the single-antenna users, and the proposed method is difficult to scale to the multi-stream communications for multi-antenna users.

The main contributions of this letter are as follows:
\begin{itemize} 
	\item We consider a more generalized mmWave downlink (DL) MU-MIMO system, where the base station (BS) with DSs simultaneously provides multi-stream communications for each user. The hybrid beamforming for the considered system is more complicated than the scenario with single-antenna users and has been less well studied. 
	
	\item We proposed a novel hybrid beamforming design for the considered system. We first assume no inter-user interference (IUI) and leverage the singular value decomposition (SVD) of user channels with the waterfilling strategy to get the optimal fully-digital beamformers and combiners. We then propose a Kuhn-Munkres (KM) algorithm aided hybrid beamforming design for DSs, which guarantees that each RF chain is connected to at least one antenna. To null out the IUI, we finally propose to project each obtained digital beamformer onto the null space of all the other equivalent user channels.
	
	\item We analyze the computational complexity of the proposed design. Numerical results demonstrate that our proposed design outperforms the conventional related methods.
\end{itemize}

\vspace{-0.4cm}
\section{System Model and Problem Formulation}
\subsection{System and Channel Model}
Consider a single-cell narrowband mmWave DL MU-MIMO system as illustrated in \figurename{\ref{Fig1}}, 
in which the BS with DSs simultaneously serves $K$ multi-antenna users. When adopting DSs, the analog beamformer is implemented by the dynamic connection network (DCN) and phase shifters, as shown in \figurename{\ref{Fig1}}. The DCN, realized by a switch network, provides the dynamic connections between the RF chains and antennas. The BS is equipped with $N_{\mathrm{T}}$ antennas and $N_{\mathrm{RF}}$ transmit RF chains. Each user is a fully-digital receiver with $N_{\mathrm{R}}$ antennas. The number of data streams for each user is $N_{s}$. Constrained by the hybrid and fully-digital beamforming structures, we have $K N_{s} \le N_{\mathrm{RF}} \le N_{\mathrm{T}}$ and $N_{s} \le N_{\mathrm{R}}$. 

\begin{figure}[!t]
	\centering
	\includegraphics[width=8.2cm, height=3.5cm]{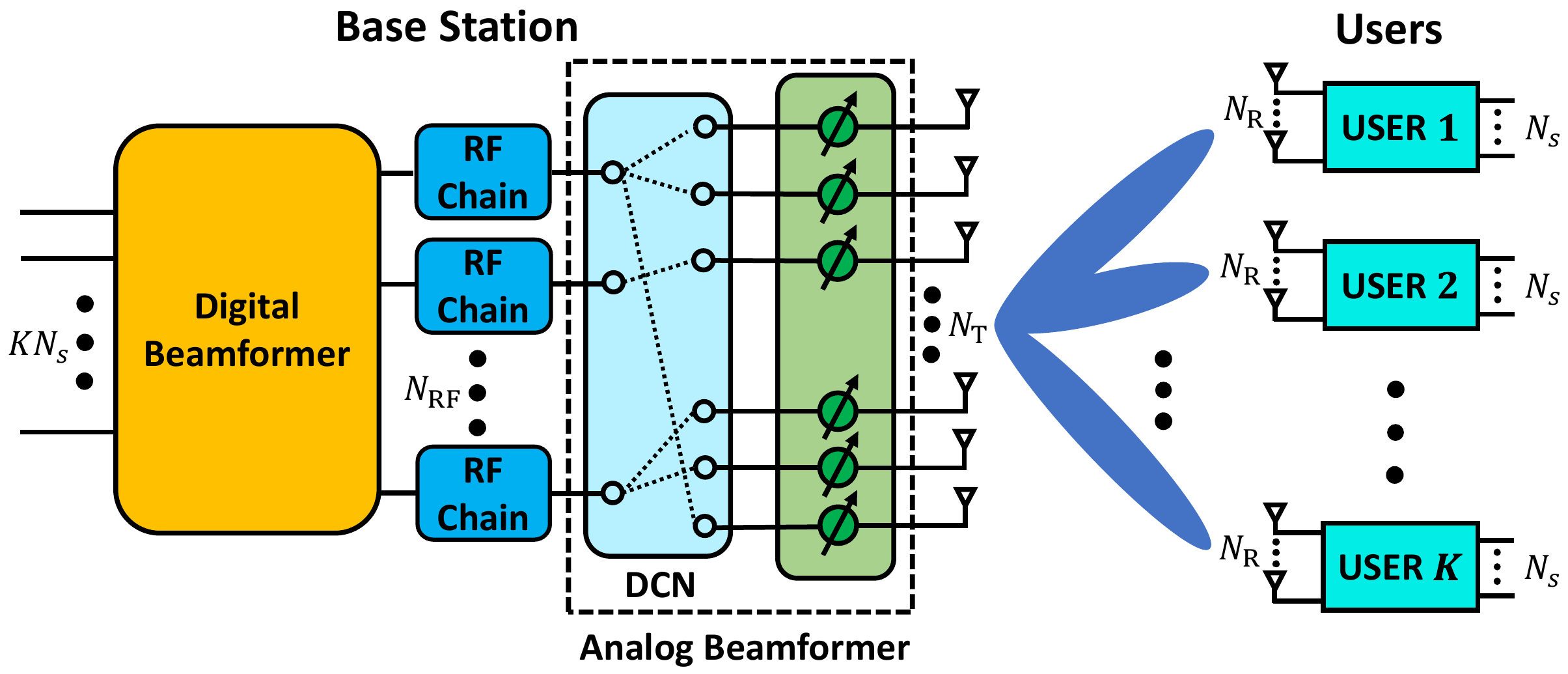}
	\vspace{-0.3cm}
	\caption{\label{Fig1} The mmWave MU-MIMO system with dynamic subarrays.}		
	\vspace{-0.55cm}
\end{figure}

At the BS, the symbol vector for $K$ users $\mathbf{s}  = [\mathbf{s}_{1}^{T}, \mathbf{s}_{2}^{T}, ..., \mathbf{s}_{K}^{T}]^{T} \in \mathbb{C}^{K N_{s} \times 1}$ satisfying $\mathbb{E} \{\mathbf{s} \mathbf{s}^{H}\} = \mathbf{I}_{K N_{s}}$, where $\mathbf{s}_{k}$ denotes the symbol vector for USER $k$, is first precoded by the digital beamformer $\mathbf{F}_{\mathrm{BB}} = [\mathbf{F}_{\mathrm{BB}1}, \mathbf{F}_{\mathrm{BB}2}, ..., \mathbf{F}_{\mathrm{BB}K}] \in \mathbb{C}^{N_{\mathrm{RF}} \times K N_{s}}$, in which $\mathbf{F}_{\mathrm{BB}k} \in \mathbb{C}^{N_{\mathrm{RF}} \times N_{s}}$ denotes the digital beamformer for the symbol vector $\mathbf{s}_{k}$, and then up-converted to the mmWave frequency through RF chains and processed by the analog beamformer $\mathbf{F}_{\mathrm{RF}} \in \mathbb{C}^{N_{\mathrm{T}} \times N_{\mathrm{RF}}}$. The final precoded signal $\mathbf{F}_{\mathrm{RF}} \mathbf{F}_{\mathrm{BB}} \mathbf{s}$ is  simultaneously transmitted from the BS to all users. 
At the receiver USER $k$, the received signal is first down-converted to the baseband and then processed by the fully-digital combiner $\mathbf{W}_{k} \in \mathbb{C}^{N_{\mathrm{R}} \times N_{s}}$. Thus, the final received baseband signal is
\begin{align*}
	\widetilde{\mathbf{y}}_{k} = & {} \mathbf{W}_{k}^{H} \mathbf{H}_{k} \mathbf{F}_{\mathrm{RF}} \mathbf{F}_{\mathrm{BB}} \mathbf{s} + \mathbf{W}_{k}^{H} \mathbf{n}_{k} \\
						   = & {} \underbrace{\mathbf{W}_{k}^{H} \mathbf{H}_{k} \mathbf{F}_{k} \mathbf{s}_{k}}_{\text{desired signal}} +  \underbrace{\mathbf{W}_{k}^{H} \mathbf{H}_{k} \sum_{i \neq k}^{K} \mathbf{F}_{i} \mathbf{s}_{i}}_{\text{inter-user interference}} + \underbrace{\mathbf{W}_{k}^{H} \mathbf{n}_{k}}_{\text{noise}},		\label{Eq1} \tag{1}
\end{align*}
where $\mathbf{H}_{k} \in \mathbb{C}^{N_{\mathrm{R}} \times N_{\mathrm{T}}}$ is the DL channel from the BS to USER $k$, $\mathbf{n}_{k} \sim \mathcal{CN}\left(\mathbf{0}_{N_{\mathrm{R}} \times 1}, \sigma^{2} \mathbf{I}_{N_{s}}\right)$ is the additive white Gaussian noise vector at USER $k$, and $\mathbf{F}_{k} = \mathbf{F}_{\mathrm{RF}} \mathbf{F}_{\mathrm{BB}k}$ denotes the overall hybrid beamformer for the symbol vector $\mathbf{s}_{k}$.   

We adopt the cluster channel model with consideration of the severe path loss \cite{mmWaveChannelModel_NYU} to characterize the narrowband mmWave channel from the BS to USER $k$, which is
\begin{align*}
	\mathbf{H}_{k} = \sqrt{\frac{N_{\mathrm{T}} N_{\mathrm{R}} \rho_{k}}{N_{c}^{k} N_{ray}^{k}}}	\sum_{c = 1}^{N_{c}^{k}} \sum_{r = 1}^{N_{ray}^{k}} \alpha_{cr}^{k} \mathbf{a}_{r}(\theta_{cr}^{k}) \mathbf{a}_{t} (\phi_{cr}^{k})^{H},		\label{Eq2} \tag{2}	
\end{align*}
where $N_{c}^{k}$ is the number of clusters in the channel from the BS to USER $k$, each of which contributes $N_{ray}^{k}$ rays, and $\rho_{k}$ is the path loss. $\alpha_{cr}^{k} \sim \mathcal{CN} (0, 1)$, $\theta_{cr}^{k}$, and $\phi_{cr}^{k}$ denote the complex gain, the angle of departure (AoD) and the angle of arrival (AoA) of the $r$-th ray in the $c$-th cluster for USER $k$, respectively. $\mathbf{a}_{r}(\cdot)$ and $\mathbf{a}_{t} (\cdot)$ are respectively the normalized receive and transmit antenna array response vectors.  
\vspace{-0.4cm}
\subsection{Problem Formulations}
In this paper, we assume that the data streams for each user follow the Gaussian distribution and the channel state information is perfectly known at the BS. Then the achievable spectral efficiency (SE) of USER $k$ is
\begin{align*}
	R_{k} = \log_{2} \vert \mathbf{I}_{N_{s}} + \mathbf{C}_{k}^{-1} \mathbf{W}_{k}^{H} \mathbf{H}_{k}  \mathbf{F}_{k} \mathbf{F}_{k}^{H} \mathbf{H}_{k}^{H} \mathbf{W}_{k} \vert,		\label{Eq3} \tag{3}
\end{align*}
where $\mathbf{C}_{k} = \mathbf{W}_{k}^{H} \mathbf{H}_{k} \left(\sum_{i \neq k}^{K} \mathbf{F}_{i} \mathbf{F}_{i}^{H}\right) \mathbf{H}_{k}^{H} \mathbf{W}_{k} + \sigma_{k}^{2} \mathbf{W}_{k}^{H} \mathbf{W}_{k}$ is the interference-plus-noise covariance matrix at USER $k$.

The problem of interest is to jointly design the dynamic hybrid beamformer at the BS and the fully-digital combiner at each user to maximize the overall SE of the mmWave MU-MIMO systems, which is formulated as
\begin{align*}
	\max_{\mathcal{A}} & {} \quad \frac{1}{K} \sum_{k = 1}^{K} R_{k},	\label{Eq4a} \tag{4a} \\
	\mathrm{s.t.} \;   & {}	\quad \sum_{k = 1}^{K} \Vert \mathbf{F}_{\mathrm{RF}} \mathbf{F}_{\mathrm{BB}k} \Vert_{\mathrm{F}}^{2} = P, \label{Eq4b} \tag{4b}	\\
					   & {} \quad \vert \mathbf{F}_{\mathrm{RF}}(i, j) \vert = 1, \quad \forall (i, j) \in \mathcal{F},					  \label{Eq4c} \tag{4c} \\
				  	   & {}	\quad \Vert \mathbf{F}_{\mathrm{RF}}(i, :) \Vert_{\mathrm{0}} = 1,    \quad i = 1, 2, ..., N_{\mathrm{T}},	  \label{Eq4d} \tag{4d}	\\
				  	   & {}	\quad \Vert \mathbf{F}_{\mathrm{RF}}(:, j) \Vert_{\mathrm{0}} \ge 1,  \; \; \; j = 1, 2, ..., N_{\mathrm{RF}},	  \label{Eq4e} \tag{4e}	
\end{align*}
where $\mathcal{A} = \{\mathbf{F}_{\mathrm{RF}}, \{\mathbf{F}_{\mathrm{BB}k}\}_{k = 1}^{K}, \{\mathbf{W}_{k}\}_{k = 1}^{K}\}$ is the set of optimization variables. \eqref{Eq4b} is the coupling transmit power constraint at the BS, where $P$ is the total transmit power. \eqref{Eq4c} denotes the unit-modulus constraint on each nonzero element of the analog beamformer $\mathbf{F}_{\mathrm{RF}}$, in which $\mathcal{F}$ is the set of nonzero elements of the analog beamformer. \eqref{Eq4d} and \eqref{Eq4e} are two constraints introduced by DSs. The former constraint guarantees that each antenna is only connected to one RF chain, and the latter constraint ensures that each RF chain connects to at least one antenna. Subject to the above constraints and the IUI, problem (4) is highly nonconvex and very difficult to solve.

\vspace{-0.5cm}
\section{Proposed Hybrid Beamforming Design}
In this section, we propose an effective three-stage hybrid beamforming design to solve the nonconvex problem (4). 
\vspace{-0.5cm}
\subsection{Fully-Digital Beamforming When Assuming No IUI}
When assuming no IUI, the transmissions from the BS to each user are independent and parallel. Thus, we get the optimal fully-digital beamformers and combiners using the SVD of the user channels and the waterfilling algorithm \cite{Tse_FoWC}. We denote by $\widetilde{\mathbf{F}}_{k}$ the fully-digital beamformer for the symbol vector $\mathbf{s}_{k}$. Let the SVD of the channel from the BS to USER $k$ be $\mathbf{H}_{k} = \mathbf{U}_{k} \mathbf{D}_{k} \mathbf{V}_{k}^{H}$, where $\mathbf{U}_{k} \in \mathbb{C}^{N_{\mathrm{R}} \times N_{\mathrm{R}}}$ and $\mathbf{V}_{k} \in \mathbb{C}^{N_{\mathrm{T}} \times N_{\mathrm{T}}}$ are unitary matrices, and $\mathbf{D}_{k} \in \mathbb{C}^{N_{\mathrm{R}} \times N_{\mathrm{T}}}$ is a rectangular diagonal matrix whose main diagonal elements are the singular values in decreasing order.
Then the optimal fully-digital beamformer and combiner are given by
\begin{align*}
	\widetilde{\mathbf{F}}_{k}^{*} = {} \mathbf{V}_{k}(:, (1:N_{s})) \mathbf{P}_{k}^{\frac{1}{2}}, \quad \mathbf{W}_{k}^{*} = {} \mathbf{U}_{k}(:, (1:N_{s})),	\label{Eq5} \tag{5}
	\vspace{-0.2cm}
\end{align*}
where $\mathbf{P}_{k}$ is a diagonal matrix whose diagonal element denotes the allocated transmit power for the corresponding data stream by the waterfilling strategy such that $\sum_{k = 1}^{K} \mathrm{Tr}\left(\mathbf{P}_{k}\right) = P$.

\vspace{-0.45cm}
\subsection{KM Algorithm Aided Hybrid Beamforming for DSs}
In mmWave MIMO systems, maximizing SE can be approximately addressed by minimizing the Frobenius norm of the gap between the optimal fully-digital and overall hybrid beamforming \cite{Heath_Overview_mmWave, JunZhang_MIMO_OFDM_HB}. Thus, assuming no IUI, the hybrid beamforming design at the BS can be further written as 
\begin{align*}
	\min_{\mathcal{B}} & {} \quad \sum_{k = 1}^{K} \Vert \widetilde{\mathbf{F}}_{k}^{*} - \mathbf{F}_{\mathrm{RF}} \mathbf{F}_{\mathrm{BB}k} \Vert_{\mathrm{F}}^{2},    \label{Eq6a} \tag{6a} \\
	\mathrm{s.t.}	   & {} \quad \Vert \mathbf{F}_{\mathrm{RF}} \mathbf{F}_{\mathrm{BB}k} \Vert_{\mathrm{F}}^{2} = P_{k},	\quad \forall k = 1, 2, ..., K,				\label{Eq6b} \tag{6b} \\
					   & {} \quad \eqref{Eq4c}-\eqref{Eq4e},
\end{align*}
where $\mathcal{B} = \{\mathbf{F}_{\mathrm{RF}}, \{\mathbf{F}_{\mathrm{BB}k}\}_{k = 1}^{K}\}$, and $P_{k} = \mathrm{Tr}\left(\mathbf{P}_{k}\right)$ such that $\sum_{k = 1}^{K} P_{k} = P$. When solving the above problem, it has been proved that we can first remove the coupling transmit power constraint \eqref{Eq6b} and finally normalize the digital beamformer to satisfy this constraint \cite{JunZhang_MIMO_OFDM_HB}. Thus we propose the following KM algorithm aided hybrid beamforming design for DSs.

\subsubsection{Analog beamformer design}  
As each antenna is only connected to one RF chain, there is only one nonzero element in each row of $\mathbf{F}_{\mathrm{RF}}$. Assuming the $i$-th antenna is connected to the $l_{i}$-th RF chain, the subproblem with respect to $\mathbf{F}_{\mathrm{RF}}$ is equivalently expressed as
\begin{align*}
	\min_{\mathcal{C}} & {} \quad \sum_{i = 1}^{N_{\mathrm{T}}} \sum_{k = 1}^{K} \Vert \widetilde{\mathbf{F}}_{k}^{*}(i,:) - \mathbf{F}_{\mathrm{RF}}(i, l_{i}) \mathbf{F}_{\mathrm{BB}k}(l_{i},:) \Vert_{\mathrm{F}}^{2},    \label{Eq7a} \tag{7a} \\
	\mathrm{s.t.}	   & {} \quad \eqref{Eq4c}-\eqref{Eq4e}.
\end{align*}
where $\mathcal{C} = \{\mathbf{F}_{\mathrm{RF}}, \{l_{i}\}_{i = 1}^{N_{\mathrm{T}}}\}$. Since there is only one nonzero element in each row of $\mathbf{F}_{\mathrm{RF}}$, we can separately update the nonzero elements of $\mathbf{F}_{\mathrm{RF}}$ row-by-row. When removing the constant terms, 
the subproblem involving the $i$-th antenna is
\begin{align*}
\hspace{-0.1cm}	\min_{\mathcal{D}_{i}} & {} \; \sum_{k = 1}^{K} \Vert \mathbf{F}_{\mathrm{BB}k}(l_{i},:) \Vert_{\mathrm{F}}^{2} - 2 \operatorname{Re}\{\mathrm{Tr}\left({A}_{i}^{l_{i}} \mathbf{F}_{\mathrm{RF}}(i, l_{i})^{H}\right)\},    \label{Eq8a} \tag{8a} \\
\hspace{-0.1cm}	\mathrm{s.t.}	   & {} \; \vert \mathbf{F}_{\mathrm{RF}}(i, l_{i})\vert = 1, \quad l_{i} \in \{1, 2, ..., N_{\mathrm{RF}}\},	\label{Eq8b} \tag{8b} \\
\hspace{-0.1cm}					   & {} \; \Vert \mathbf{F}_{\mathrm{RF}}(i, :) \Vert_{\mathrm{0}} = 1,  \label{Eq8c} \tag{8c}
\end{align*}
where $\mathcal{D}_{i} = \{\mathbf{F}_{\mathrm{RF}}(i, l_{i}), l_{i}\}$ and $A_{i}^{l_{i}} = \sum_{k = 1}^{K} \widetilde{\mathbf{F}}_{k}^{*}(i,:) \mathbf{F}_{\mathrm{BB}k}^{H}(l_{i},:)$. 
Thus, the optimal solutions to minimize \eqref{Eq8a} are
\begin{align*}
	  \hspace{-0.4cm} l_{i}^{*} \hspace{-0.05cm}  = \hspace{-0.05cm} \min_{l_{i}} \sum_{k = 1}^{K} \Vert \mathbf{F}_{\mathrm{BB}k}(l_{i},:) \Vert_{\mathrm{F}}^{2} - 2 \vert A_{i}^{l_{i}} \vert,  
      \mathbf{F}_{\mathrm{RF}}^{*}(i, l_{i}^{*}) \hspace{-0.05cm} = \hspace{-0.05cm} \frac{A_{i}^{l_{i}^{*}}}{\vert A_{i}^{l_{i}^{*}} \vert}. \hspace{-0.35cm} \label{Eq9} \tag{9}	
\end{align*}

As constraint \eqref{Eq4e} is not considered in problem (8), the obtained solutions for each antenna using \eqref{Eq9} may not satisfy this constraint. In this case, we reallocate some antennas to the RF chains connected to no antennas. To reallocate antennas, two key questions must be addressed: 
\textit{how many antennas are required to be reallocated} and \textit{how to select and reallocate the required number of antennas}. For the first question, we give the following proposition.
\begin{proposition}
	\vspace{-0.1cm}
	\label{Prop1}
	Let $N_{\mathrm{RF}}^{0}$ denote the number of RF chains connected to no antennas, and then the optimal number of antennas to be reallocated is $N_{\mathrm{RF}}^{0}$.
\end{proposition}
\begin{proof}
	Let $f_{i}(\mathcal{D}_{i}^{l_{i}})\hspace{-0.1cm} = \hspace{-0.1cm} \sum_{k = 1}^{K} \hspace{-0.05cm} \Vert \widetilde{\mathbf{F}}_{k}^{*}(i,:) \hspace{-0.05cm} - \hspace{-0.05cm} \mathbf{F}_{\mathrm{RF}}(i, l_{i}) \mathbf{F}_{\mathrm{BB}k}(l_{i},:) \Vert_{\mathrm{F}}^{2}$, where $\mathcal{D}_{i}^{l_{i}} = \{\mathbf{F}_{\mathrm{RF}}^{*}(i, l_{i}), l_{i}\}$, $i = 1, 2, ..., N_{\mathrm{T}}$, denotes the $i$-th antenna is connected to the $l_{i}$-th RF chain with the corresponding optimal analog beamformer element $\mathbf{F}_{\mathrm{RF}}^{*}(i, l_{i}) = \frac{A_{i}^{l_{i}}}{\vert A_{i}^{l_{i}} \vert}$. Thus for the $i$-th antenna, the obtained solution $\mathcal{D}_{i}^{l_{i}^{*}}$ using \eqref{Eq9} can minimize $f_{i}(\mathcal{D}_{i}^{l_{i}})$. Assuming that the $j$-th antenna is reallocated from the $l_{j}^{*}$-th RF chain to the $l_{j}^{\prime}$-th RF chain, we get $f_{j}(\mathcal{D}_{j}^{l_{j}^{\prime}}) \ge f_{j}(\mathcal{D}_{j}^{l_{j}^{*}})$ such that $\sum_{i \ne j}^{N_{\mathrm{T}}} f_{i}(\mathcal{D}_{i}^{l_{i}^{*}}) +  f_{j}(\mathcal{D}_{j}^{l_{i}^{\prime}}) \ge \sum_{i \ne j}^{N_{\mathrm{T}}} f_{i}(\mathcal{D}_{i}^{l_{i}^{*}}) +  f_{j}(\mathcal{D}_{j}^{l_{i}^{*}})$. Thus, we can infer that the more antennas are reallocated, the larger the objective function \eqref{Eq6a} is. To satisfy the constraint \eqref{Eq4e}, the number of reallocated antennas can not be fewer than $N_{\mathrm{RF}}^{0}$. Therefore, the optimal number of antennas to be reallocated is $N_{\mathrm{RF}}^{0}$.
	\vspace{-0.1cm}
\end{proof}
Based on \proposname{\ref{Prop1}}, the second question becomes \textit{how to select $N_{\mathrm{RF}}^{0}$ antennas and reallocate them to $N_{\mathrm{RF}}^{0}$ RF chains}. We denote by $\mathcal{S}_{\mathrm{RA}}$ the set of reallocated antennas thus $\vert \mathcal{S}_{\mathrm{RA}} \vert = N_{\mathrm{RF}}^{0}$. Let $\mathcal{S}_{\mathrm{RF}0}$ be the set of RF chains connected to no antennas. For $i \in \mathcal{S}_{\mathrm{RA}}$, we assume that the $i$-th antenna is only connected to the $l_{i}$-th RF chain, where $l_{i} \in \mathcal{S}_{\mathrm{RF}0}$. Then the objective function \eqref{Eq6a} after the antenna reallocation is 
\begin{align*}
   & {} \sum_{i \notin \mathcal{S}_{\mathrm{RA}}} f_{i}(\mathcal{D}_{i}^{l_{i}^{*}}) + \sum_{i \in \mathcal{S}_{\mathrm{RA}}} f_{i}(\mathcal{D}_{i}^{l_{i}}) \\
 = & {} \sum_{i = 1}^{N_{\mathrm{T}}} f_{i}(\mathcal{D}_{i}^{l_{i}^{*}}) + \sum_{i \in \mathcal{S}_{\mathrm{RA}}} (f_{i}(\mathcal{D}_{i}^{l_{i}}) - f_{i}(\mathcal{D}_{i}^{l_{i}^{*}})). 		\label{Eq10} \tag{10}
\end{align*}
When not considering constraint \eqref{Eq4e}, the obtained solutions $\{\mathcal{D}_{i}^{l_{i}^{*}}\}_{i = 1}^{N_{\mathrm{T}}}$ achieve the minimum of \eqref{Eq6a}, which is the first term in \eqref{Eq10}. Hence, to minimize the objective function \eqref{Eq6a} after the antenna reallocation, we should minimize the second term in \eqref{Eq10}, which  denotes the total increase in \eqref{Eq6a} caused by the antenna reallocation. To address the problem, we first construct the \textit{reallocation cost matrix} $\mathbf{G} \in \mathbb{C}^{M \times N_{\mathrm{RF}}^{0}}$, where $M$ is the number of antennas that can be reallocated. Let $\mathcal{S}_{l}, l = 1, 2, ..., N_{\mathrm{RF}}$, denote the set of antennas connected to the $l$-th RF chain. For the $l$-th RF chain, if it only connects to one antenna, this antenna will not be reallocated to prevent the $l$-th RF chain from being probably connected to no antennas after the antenna reallocation. Thus the set of antennas that can be reallocated is defined as $\mathcal{S}_{\mathrm{CRA}} = \cup_{l} \mathcal{S}_{l},  \text{ where }  \vert \mathcal{S}_{l} \vert > 1$,
and thus $M = \vert \mathcal{S}_{\mathrm{CRA}} \vert$.  $\forall a_{i} \in \mathcal{S}_{\mathrm{CRA}}, l_{j} \in \mathcal{S}_{\mathrm{RF}0}$, where $i$ and $j$ are respectively their positions in $\mathcal{S}_{\mathrm{CRA}}$ and $\mathcal{S}_{\mathrm{RF}0}$, we define the corresponding reallocation cost matrix element $\mathbf{G}(i,j)$ as
\begin{align*}
	\mathbf{G}(i,j) = f_{a_{i}}(\mathcal{D}_{a_{i}}^{l_{j}}) - f_{a_{i}}(\mathcal{D}_{a_{i}}^{l_{a_{i}}^{*}}),	\label{Eq11} \tag{11}
\end{align*}
which denotes the increase in the objective function \eqref{Eq6a} after reallocating the $a_{i}$-th antenna from the optimal $l_{a_{i}}^{*}$-th RF chain to the $l_{j}$-th RF chain. From the proof in \proposname{\ref{Prop1}}, we get $\mathbf{G}(i,j) \ge 0$. Thus the antenna reallocation to minimize the second term in \eqref{Eq10} can be further expressed as

\textit{Given the reallocation cost matrix $\mathbf{G}$, where $\mathbf{G}(i, j) \ge 0$, $\forall i, j$, select one element in each column so that these selected elements are in different rows and their sum is minimum.}

This is a classical unbalanced assignment problem, which can be efficiently solved by the extended KM algorithm \cite{ExMunkres_Assignment}. If $\mathbf{G}(m,n)$ is selected, we reallocate the $a_{m}$-th antenna from the $l_{a_{m}}^{*}$-th RF chain to the $l_{n}$-th RF chain. There is a special case where all the antennas connected to one RF chain are reallocated. If this happens, we keep the antenna that induces the maximum reallocation cost connected to this RF chain, remove the corresponding row in $\mathbf{G}$, and then reallocate the rest antennas again. Hence the proposed analog beamforming can guarantee that each RF chain is connected to at least one antenna. The details of our proposed KM aided dynamic analog beamforming is summarized in \algorithmname{\ref{Algo1}}.

\begin{figure}[!t]
	\removelatexerror
	
	\vspace{-9pt}
	\begin{algorithm}[H]
		
		\caption{KM Aided Dynamic Analog Beamformer Design\label{Algo1}}
		\begin{algorithmic}[1]
			\State \Input 	$\{\widetilde{\mathbf{F}}_{k}^{*}\}_{k = 1}^{K}$, $\{\mathbf{F}_{\mathrm{BB}k}\}_{k = 1}^{K}$, $N_{\mathrm{T}}$, $N_{\mathrm{RF}}$
			\State \Output	$\mathbf{F}_{\mathrm{RF}}$
			\State Initialize $\mathcal{S}_{1} = \mathcal{S}_{2} = ... = \mathcal{S}_{N_{\mathrm{RF}}} = \varnothing$, $\mathcal{S}_{\mathrm{RF}0} = \varnothing$, $\mathcal{S}_{\mathrm{CRA}} = \varnothing$, $\mathbf{F}_{\mathrm{RF}}^{*} = \mathbf{0}_{N_{\mathrm{T}} \times N_{\mathrm{RF}}}$, $i = 1$, $j = 1$, $N_{\mathrm{RF}}^{0} = 0$, $M = 0$.
			\Repeat
			\State Get $l_{i}^{*}$ and $\mathbf{F}_{\mathrm{RF}}^{*}(i, l_{i}^{*})$ according to \eqref{Eq9}, 
			$\mathcal{S}_{l_{i}^{*}} = \mathcal{S}_{l_{i}^{*}} \cup {i}$,  $i = i + 1$;
			\Until{($i > N_{\mathrm{T}}$)}
			
			\Repeat
			\State \IF{$\vert \mathcal{S}_{j} \vert == 0$} \THEN
			$\mathcal{S}_{\mathrm{RF}0} = \mathcal{S}_{\mathrm{RF}0} \cup j$, $N_{\mathrm{RF}}^{0} = N_{\mathrm{RF}}^{0} + 1$;
			\State \ELSEIF{$\vert \mathcal{S}_{j} \vert > 1$} \THEN
			$\mathcal{S}_{\mathrm{CRA}} = \mathcal{S}_{\mathrm{CRA}} \cup \mathcal{S}_{j}$, $M = M + 1$;
			\ENDIF
			\State $j = j + 1$;
			\Until{($j > N_{\mathrm{RF}}$)}
			\If{$N_{\mathrm{RF}}^{0} > 0$}
			\State Construct the \textit{reallocation cost matrix} $\mathbf{G}$ according to \eqref{Eq11}; 
			
			\Repeat
			\State Select one element in each column of $\mathbf{G}$ using the extended KM algorithm. Define $\mathcal{Q} = \{(x_{i}, y_{i}) | 1 \le x_{i} \le M, 1 \le y_{i} \le  N_{\mathrm{RF}}^{0}, a_{x_{i}} \in \mathcal{S}_{\mathrm{CRA}}, l_{y_{i}} \in \mathcal{S}_{\mathrm{RF}0}, i = 1, 2, ...,  N_{\mathrm{RF}}^{0} \}$ as the set of positions of the selected elements in $\mathbf{G}$, $q = 1$, $flag = 0$, $\mathbf{F}_{\mathrm{RF}} = \mathbf{F}_{\mathrm{RF}}^{*}$, $\mathbf{m} \in \mathbf{0}_{N_{\mathrm{RF}}^{0} \times 1}$, and $\mathbf{z} \in \mathbf{0}_{N_{\mathrm{RF}}^{0} \times 1}$; 
			
			\State Let $\mathcal{S}_{n}^{\prime} = \mathcal{S}_{n}$, $\forall n = 1, 2, ...,  N_{\mathrm{RF}}^{0}$;
			
			\Repeat 
			\State Get $a_{x_{q}}$ from $\mathcal{S}_{\mathrm{CRA}}$ and $l_{y_{q}}$ from $\mathcal{S}_{\mathrm{RF}0}$, and let $n = 1$;
			\Repeat
			\If{$a_{x_{q}} \in \mathcal{S}_{n}^{\prime}$} 
			\State	$\mathcal{S}_{n}^{\prime} = \mathcal{S}_{n}^{\prime} \setminus a_{x_{q}}$,
			$\mathcal{S}_{l_{y_{q}}}^{\prime} = \mathcal{S}_{l_{y_{q}}}^{\prime} \cup a_{x_{q}}$;
			\State 	$\mathbf{F}_{\mathrm{RF}}(a_{x_{q}}, l_{a_{x_{q}}}^{*}) = 0$,
			$\mathbf{F}_{\mathrm{RF}}(a_{x_{q}}, l_{y_{q}}) = \frac{A_{a_{x_{q}}}^{l_{y_{q}}}}{\vert A_{a_{x_{q}}}^{l_{y_{q}}} \vert}$;
			\If{$\mathbf{m}(n) < \mathbf{G}(x_{q}, y_{q})$}
			\State $\mathbf{m}(n) = \mathbf{G}(x_{q}, y_{q})$, $\mathbf{z}(n) = x_{q}$;
			\EndIf
			\State \IF{$\vert \mathcal{S}_{n}^{\prime} \vert$ == 0} \THEN
			$flag = 1$;
			\ENDIF
			\State \Break;
			\Else
			\State $n = n + 1$;
			\EndIf	
			\Until{($n > N_{\mathrm{RF}}^{0}$)}
			\If{$flag == 1$}
			\State $\mathbf{G} = [\mathbf{G}(1:(\mathbf{z}(n)-1), :)^{T}, \mathbf{G}((\mathbf{z}(n)+1):M, :)^{T}]^{T}$;	
			\State $\mathcal{S}_{\mathrm{CRA}} = \mathcal{S}_{\mathrm{CRA}} \setminus a_{\mathbf{z}(n)}$, $M = \vert \mathcal{S}_{\mathrm{CRA}} \vert$;
			\State \Break;
			\EndIf
			\State $q = q + 1$;
			\Until{($flag == 0$ and $q > N_{\mathrm{RF}}^{0}$)}
			\Until{($flag == 0$)}
			\Else
			\State $\mathbf{F}_{\mathrm{RF}} = \mathbf{F}_{\mathrm{RF}}^{*}$;
			\EndIf
			\State \Return $\mathbf{F}_{\mathrm{RF}}$.
		\end{algorithmic}
	\end{algorithm}
	\vspace{-1cm}
\end{figure}

\subsubsection{Digital beamformer design}
When the analog beamformer is fixed and the constant terms are removed, the subproblem with respect to the digital beamformer $\mathbf{F}_{\mathrm{BB}k}$ is
\begin{align*}
\hspace{-0.25cm} \min_{\mathbf{F}_{\mathrm{BB}k}} \hspace{-0.05cm} \mathrm{Tr} (\mathbf{F}_{\mathrm{BB}k}^{H} \mathbf{F}_{\mathrm{RF}}^{H} \mathbf{F}_{\mathrm{RF}} \mathbf{F}_{\mathrm{BB}k}) \hspace{-0.05cm} - \hspace{-0.05cm} 2 \operatorname{Re}\{\mathrm{Tr}(\mathbf{F}_{\mathrm{BB}k}^{H} \mathbf{F}_{\mathrm{RF}}^{H} \widetilde{\mathbf{F}}_{k}^{*})\}. \hspace{-0.2cm}	\label{Eq12} \tag{12}
\end{align*}
Checking the first-order optimality condition for $\mathbf{F}_{\mathrm{BB}k}$ yields
\begin{align*}
	\mathbf{F}_{\mathrm{BB}k}^{*} = (\mathbf{F}_{\mathrm{RF}}^{H} \mathbf{F}_{\mathrm{RF}})^{-1} \mathbf{F}_{\mathrm{RF}}^{H} \widetilde{\mathbf{F}}_{k}^{*}. \label{Eq13} \tag{13}
\end{align*}

Update the analog and digital beamformers alternatively until some termination condition triggers. Then we normalize $\mathbf{F}_{\mathrm{BB}k}$ by a factor of $\frac{\sqrt{P_{k}}}{\Vert \mathbf{F}_{\mathrm{RF}} \mathbf{F}_{\mathrm{BB}k} \Vert_{\mathrm{F}}}$ to satisfy the constraint \eqref{Eq6b}.

\vspace{-0.4cm}
\subsection{Null Space Projection to Cancel the IUI}
As the assumption that there is no IUI usually cannot be satisfied, we resort to the null space projection (NSP) method to null out the IUI. For USER $k$, 
we define $\widetilde{\mathbf{H}}_{\mathrm{eq}k}$ as
\begin{align*}
	\widetilde{\mathbf{H}}_{\mathrm{eq}k} = [\mathbf{H}_{\mathrm{eq}1}^{T}, ..., \mathbf{H}_{\mathrm{eq}(k-1)}^{T}, \mathbf{H}_{\mathrm{eq}(k+1)}^{T}, ..., \mathbf{H}_{\mathrm{eq}K}^{T}]^{T},		\label{Eq14} \tag{14}
\end{align*}
where $\mathbf{H}_{\mathrm{eq}i} = \mathbf{W}_{i}^{H} \mathbf{H}_{i} \mathbf{F}_{\mathrm{RF}}$ is the equivalent channel for USER $i$. To cancel the IUI, we project the obtained digital beamformer $\mathbf{F}_{\mathrm{BB}k}$ onto the null space of $\widetilde{\mathbf{H}}_{\mathrm{eq}k}$. Let $\mathbf{C}_{k} = \mathcal{N}(\widetilde{\mathbf{H}}_{\mathrm{eq}k})$ be the orthogonal basis for the null space of $\widetilde{\mathbf{H}}_{\mathrm{eq}k}$, and then its projection matrix is $\mathbf{Q}_{k} = \mathbf{C}_{k} (\mathbf{C}_{k}^{H} \mathbf{C}_{k})^{-1} \mathbf{C}_{k}^{H}$.
Thus, the obtained digital beamformer after NSP is
\begin{align*}
	\mathbf{F}_{\mathrm{BB}k} = \mathbf{Q}_{k} \mathbf{F}_{\mathrm{BB}k},			\label{Eq15} \tag{15}
\end{align*}
which is further normalized by a factor of $\frac{\sqrt{P_{k}}}{\Vert \mathbf{F}_{\mathrm{RF}} \mathbf{F}_{\mathrm{BB}k} \Vert_{\mathrm{F}}}$. We summarize our proposed three-stage design in \algorithmname{\ref{Algo2}}.

We simply analyze the complexity of \algorithmname{\ref{Algo2}}. In the first stage, the complexity is dominated by the SVD of the user channels, which is $\mathcal{O}(K(N_{\mathrm{T}} N_{\mathrm{R}}^{2} + N_{\mathrm{T}}^{2} N_{\mathrm{R}}))$. 
In the second stage, the complexity of updating the analog and digital beamformers is respectively $\mathcal{O}(K N_{\mathrm{T}} N_{\mathrm{RF}} N_{s} + I_{1} I_{2} M (N_{\mathrm{RF}}^{0})^{2})$ and $\mathcal{O}(K I_{1} N_{\mathrm{T}} N_{\mathrm{RF}}^{2})$, where $I_{1}$ and $I_{2}$ denote the number of iterations required by the second stage and  \algorithmname{\ref{Algo1}}. 
In the third stage, the complexity of NSP to cancel the IUI is $\mathcal{O}(K (N_{\mathrm{T}} N_{\mathrm{RF}} N_{s} + N_{\mathrm{RF}}^{3}))$. 
Thus, the overall complexity is $\mathcal{O}(K(N_{\mathrm{T}} N_{\mathrm{R}}^{2} + N_{\mathrm{T}}^{2} N_{\mathrm{R}} + I_{1} N_{\mathrm{T}} N_{\mathrm{RF}}^{2}) + I_{1} I_{2} M (N_{\mathrm{RF}}^{0})^{2})$. 

\normalem
\begin{figure}[t]
	\removelatexerror
	
	\vspace{-9pt}
	\begin{algorithm}[H]
		
		\caption{Proposed Hybrid Beamforming Design for mmWave MU-MIMO Systems with Dynamic Subarrays \label{Algo2}}
		
		\begin{algorithmic}[1]
			\State \Input 	$\{\mathbf{H}_{k}\}_{k = 1}^{K}$, $N_{\mathrm{T}}$, $N_{\mathrm{R}}$, $N_{\mathrm{RF}}$, $N_{\mathrm{s}}$, $P$ 
			\State \Output  $\{\mathbf{W}_{k}\}_{k = 1}^{K}$, $\mathbf{F}_{\mathrm{RF}}$, $\{\mathbf{F}_{\mathrm{BB}k}\}_{k = 1}^{K}$
			\State Define the tolerance of accuracy $\varepsilon$ and the maximum number of iterations $L$. Initialize $k = 1$. 
			\Repeat 	\Comment{the first stage}
				 \State Get $\widetilde{\mathbf{F}}_{k}^{*}$ and $\mathbf{W}_{k}$ according to \eqref{Eq5}, $P_{k} = \mathrm{Tr}\left(\mathbf{P}_{k}\right)$, $k = k + 1$;
			\Until{($k > K$)}
			\State Initialize $\mathbf{F}_{\mathrm{RF}}$ randomly satisfying constraints \eqref{Eq4c}-\eqref{Eq4e}, $\mathbf{F}_{\mathrm{BB}k} = (\mathbf{F}_{\mathrm{RF}}^{H} \mathbf{F}_{\mathrm{RF}})^{-1} \mathbf{F}_{\mathrm{RF}}^{H} \widetilde{\mathbf{F}}_{k}^{*}$, $k = 1, 2, ..., K$, and $l = 1$;
			\Repeat		\Comment{the second stage}
				\State Update $\mathbf{F}_{\mathrm{RF}}$ using \algorithmname{\ref{Algo1}};
				\State $\Delta_{1} = \sum_{k = 1}^{K} \Vert \widetilde{\mathbf{F}}_{k}^{*} - \mathbf{F}_{\mathrm{RF}} \mathbf{F}_{\mathrm{BB}k} \Vert_{\mathrm{F}}^{2}$;
				\State $\mathbf{F}_{\mathrm{BB}k} = (\mathbf{F}_{\mathrm{RF}}^{H} \mathbf{F}_{\mathrm{RF}})^{-1} \mathbf{F}_{\mathrm{RF}}^{H} \widetilde{\mathbf{F}}_{k}^{*}$, $k = 1, 2, ..., K$;
				\State $\Delta_{2} = \sum_{k = 1}^{K} \Vert \widetilde{\mathbf{F}}_{k}^{*} - \mathbf{F}_{\mathrm{RF}} \mathbf{F}_{\mathrm{BB}k} \Vert_{\mathrm{F}}^{2}$, $l = l + 1$;
			\Until{($l > L$ or $\vert \Delta_{1} - \Delta_{2}\vert < \varepsilon$)}
			
			\State $\mathbf{F}_{\mathrm{BB}k} = \frac{\sqrt{P_{k}}}{\Vert \mathbf{F}_{\mathrm{RF}} \mathbf{F}_{\mathrm{BB}k} \Vert_{\mathrm{F}}}\mathbf{F}_{\mathrm{BB}k}$, $\forall k = 1, 2, ..., K$. Let $k = 1$;
			\Repeat		\Comment{the third stage}
				\State Construct $\widetilde{\mathbf{H}}_{\mathrm{eq}k}$ according to \eqref{Eq14};
				\State $\mathbf{C}_{k} = \mathcal{N}(\widetilde{\mathbf{H}}_{\mathrm{eq}k})$, $\mathbf{Q}_{k} = \mathbf{C}_{k} (\mathbf{C}_{k}^{H} \mathbf{C}_{k})^{-1} \mathbf{C}_{k}^{H}$;
				\State $\mathbf{F}_{\mathrm{BB}k} = \mathbf{Q}_{k} \mathbf{F}_{\mathrm{BB}k}$, $\mathbf{F}_{\mathrm{BB}k} = \frac{\sqrt{P_{k}}}{\Vert \mathbf{F}_{\mathrm{RF}} \mathbf{F}_{\mathrm{BB}k} \Vert_{\mathrm{F}}}\mathbf{F}_{\mathrm{BB}k}$, $k = k + 1$;
			\Until{($k > K$)}
			\State \Return $\{\mathbf{W}_{k}\}_{k = 1}^{K}$, $\mathbf{F}_{\mathrm{RF}}$, $\{\mathbf{F}_{\mathrm{BB}k}\}_{k = 1}^{K}$.
 		\end{algorithmic}
		
	\end{algorithm}
	\vspace{-0.66cm}
\end{figure}
\ULforem
\vspace{-0.35cm}
\section{Numerical Results}
\vspace{-0.1cm}

\begin{figure}[!t]
	\centering
	\includegraphics[width=5cm, height=3.8cm]{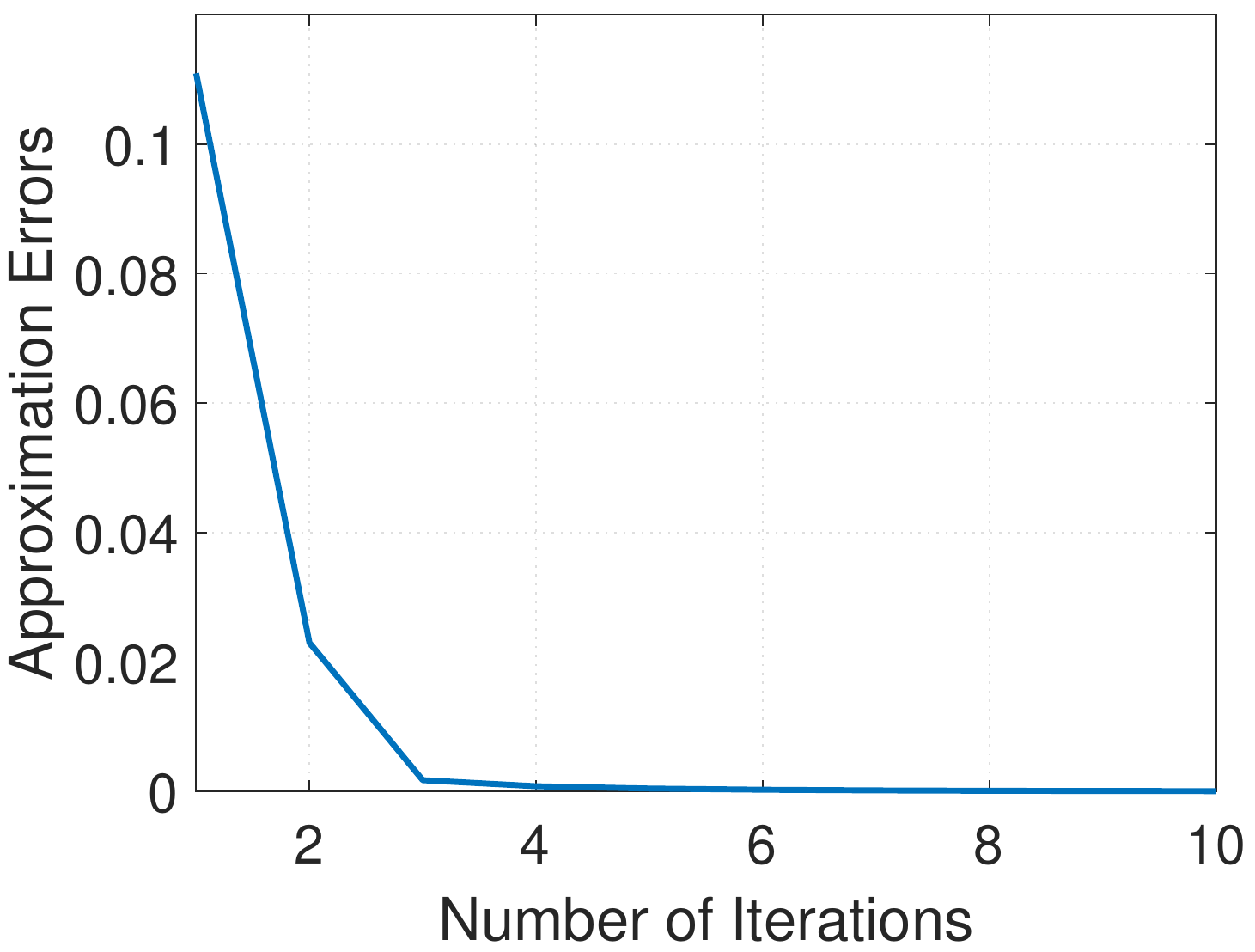}  
	\vspace{-0.35cm}
	\caption{\label{Fig2} Convergence of the second stage in the proposed design.}
	\vspace{-0.7cm}
\end{figure}

\vspace{-0.05cm}
In this section, we present numerical results to evaluate the performance of the proposed hybrid beamforming design for mmWave MU-MIMO systems with DSs.  
For comparison, the classical hybrid beamforming for DSs, i.e., \textbf{Mink-L1}\cite{Heath_SubA_HB}, \textbf{DHBF}\footnote{For fair comparisons, we only consider the fully-digital beamforming obtained by the SVD of user channels as in \eqref{Eq5} instead of the PDD method.}\cite{mmWaveWB_Dynamic_DUT}, and \textbf{LowCom}\cite{Dynamic_LowCom}, and that for FSs, i.e., \textbf{SDR-AltMin}\cite{JunZhang_MIMO_OFDM_HB}, are used in the second stage as benchmarks.
The hybrid beamforming for mmWave MU-MIMO systems with FCSs in \cite{HBF_MUMIMO_BUPT}, termed \textbf{HBF-FCS}, and the fully-digital beamforming assuming no IUI, named \textbf{FD}, are both considered.

In the simulations, the carrier frequency is assumed to be $28 \text{ GHz}$. 
The cell radius is $40 \text{ m}$, and the path loss for each user is referred to \cite{mmWaveChannelModel_NYU}.  
For user channel $\mathbf{H}_{k}$, $k = 1, 2, ..., K$, the number of clusters and rays is respectively set as $N_{c}^{k} = 6$ and $N_{ray}^{k} = 15$. The AoAs/AoDs of rays in the same cluster follow the Laplacian distribution with the mean angle uniformly distributed in $[0, 2\pi]$ and the angular spread of $10^{\degree}$. We consider the uniform linear arrays, and the antenna spacing is a half wavelength. 
Unless otherwise specified, the BS is equipped with $N_{\mathrm{T}} = 64$ antennas and $N_{\mathrm{RF}} = 16$ RF chains and each user with $N_{\mathrm{R}} = 4$ antennas. The number of users is set as $K = 6$, and that of the data streams for each user is $N_{s} = 2$. The total transmit power is $P = 30 \text{ dBm}$, and signal-to-noise ratio (SNR) is defined as $\mathrm{SNR} = \frac{P_{rx}}{\sigma^{2}}$, where $P_{rx}$ is the received power. 
We set the tolerance of accuracy $\varepsilon = 10^{-4}$ and the maximum number of iterations $L = 200$. Simulation results are averaged over 500 random channel realizations.

\begin{figure*}[!t] 
	\centering  
	\subfigcapskip=-5pt
	\subfigure[]{
		\label{Fig3a}
		\includegraphics[width=0.31 \linewidth]{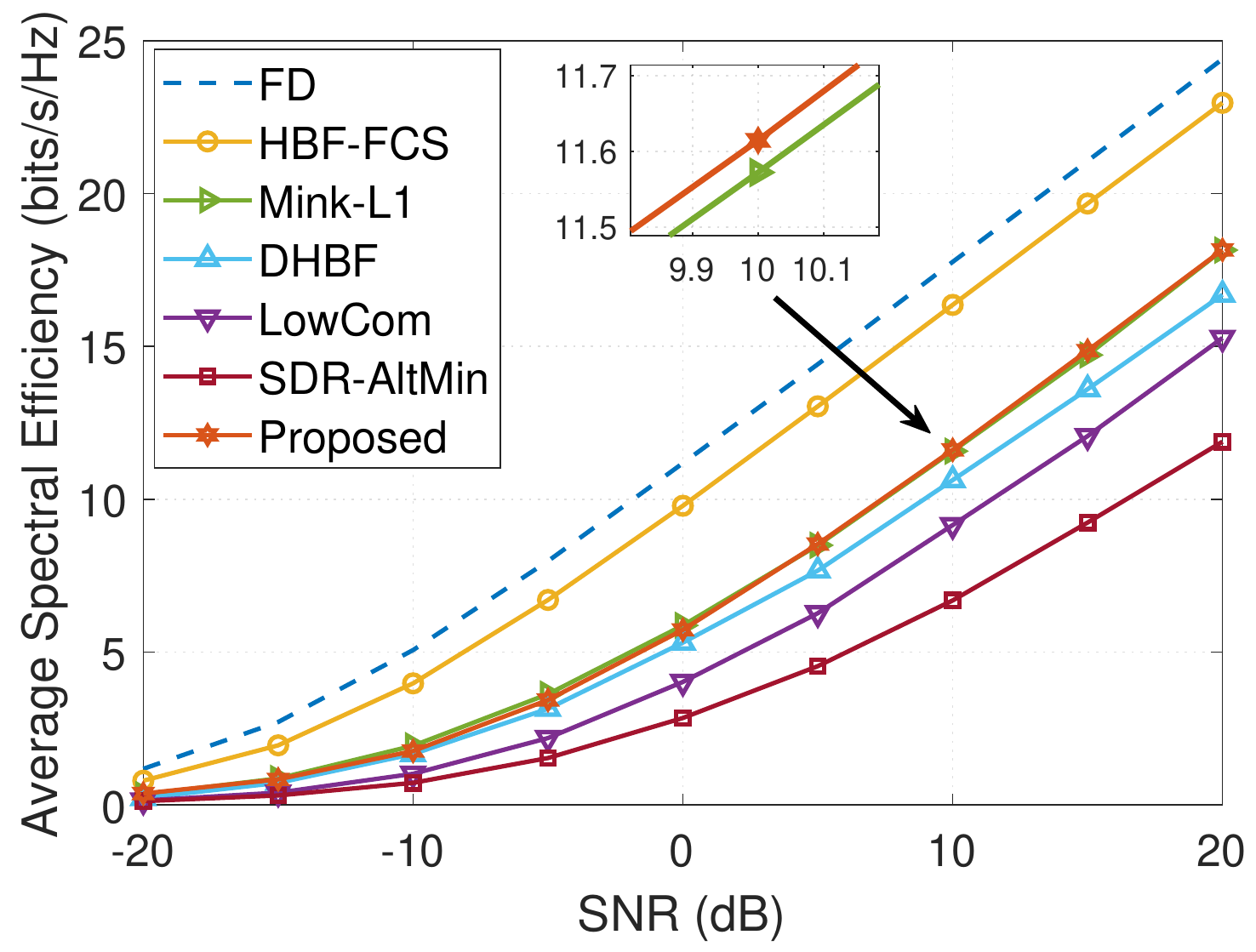}} 
	\subfigure[]{
		\label{Fig3b}
		\includegraphics[width=0.31 \linewidth]{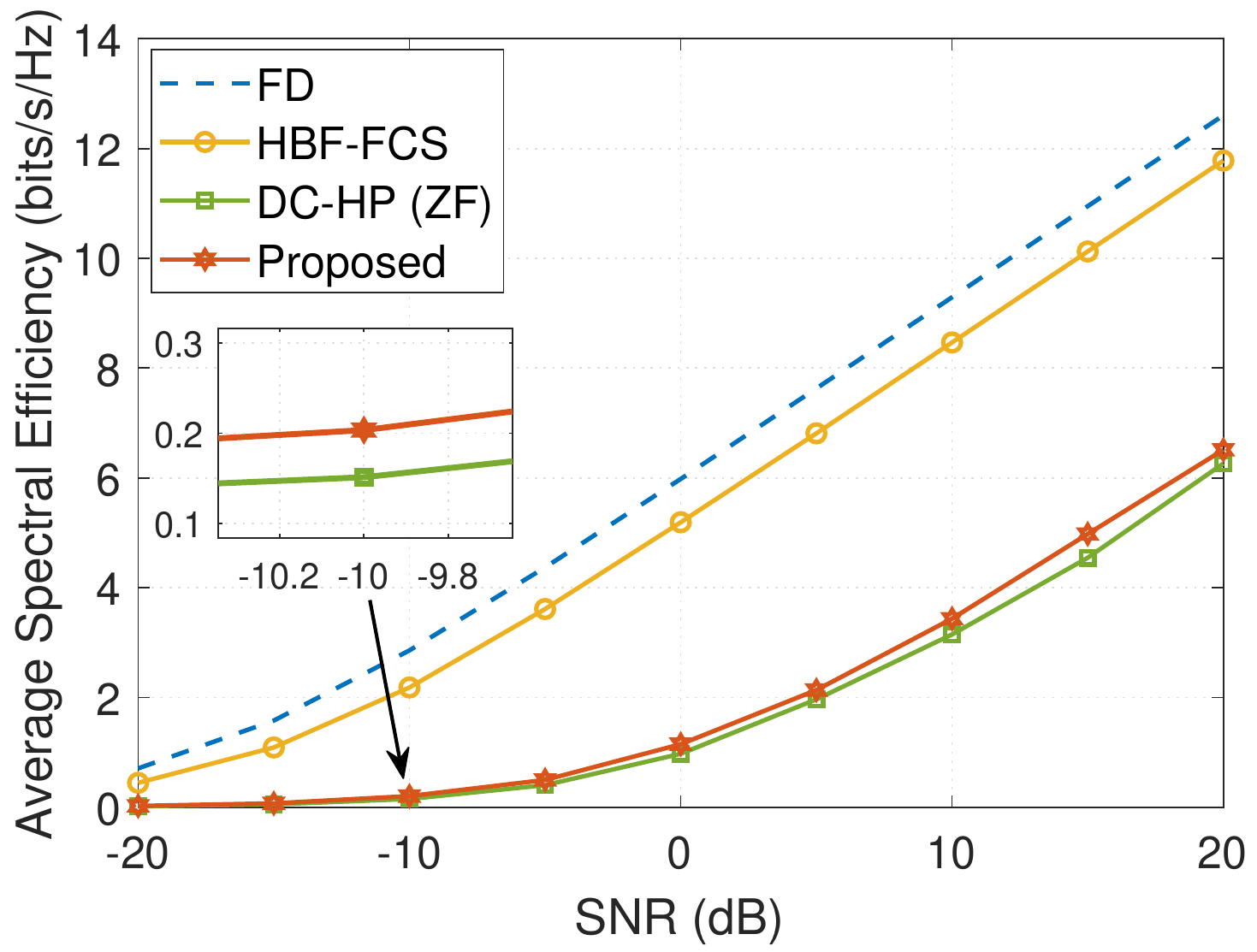}}
	\subfigure[]{
		\label{Fig3c}
		\includegraphics[width=0.31 \linewidth]{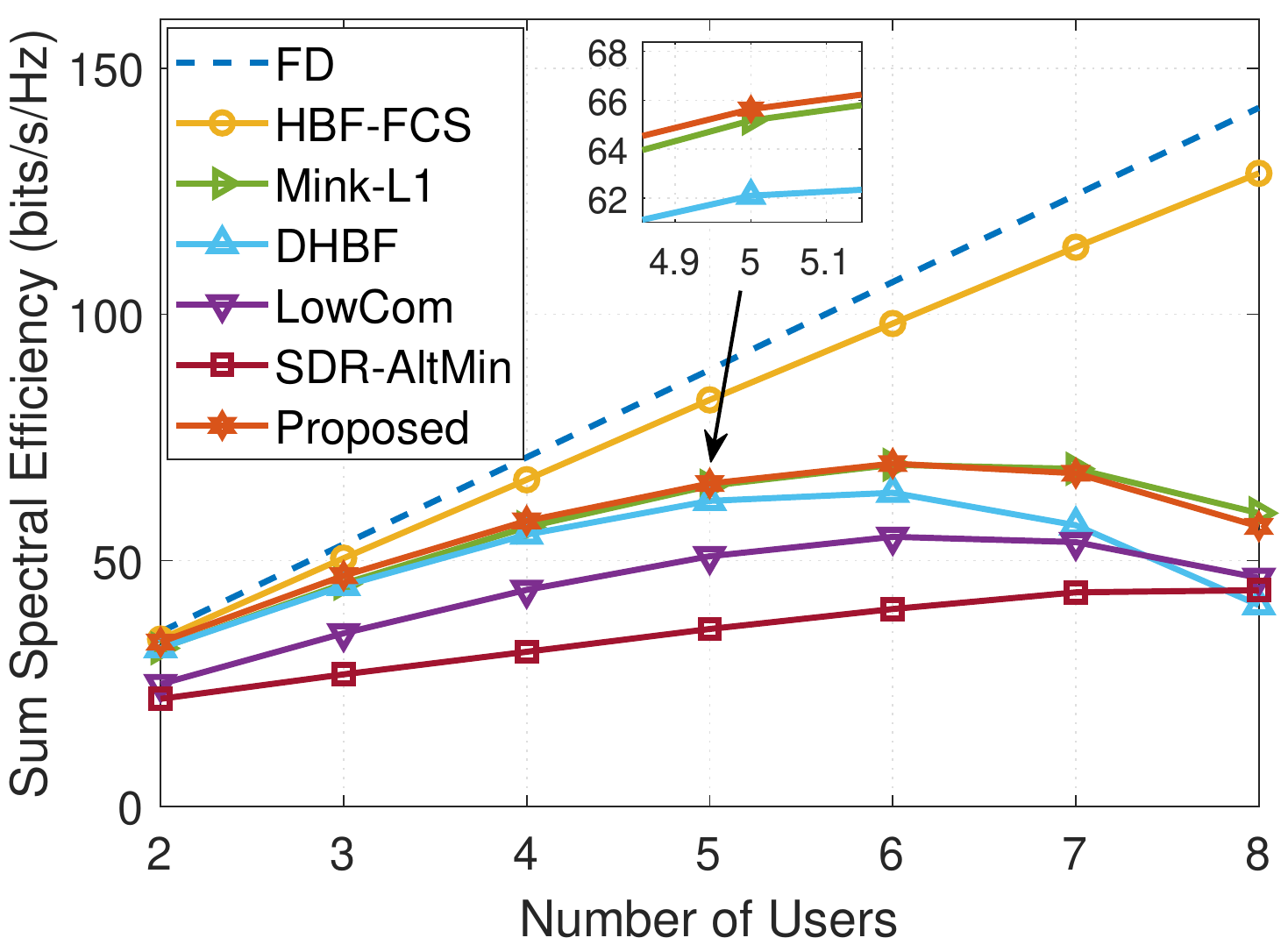}} 
	\vspace{-0.6em}
	\caption{\label{Fig3} The spectral efficiency performance of the proposed design. (a) The average spectral efficiency of each multi-antenna user versus SNR. (b) The average spectral efficiency of each single-antenna user versus SNR. (c) The sum spectral efficiency of all users versus the number of users.}
	\vspace{-1.8em}
\end{figure*}

\figurename{\ref{Fig2}} shows the convergence of the second stage in our proposed dynamic hybrid beamforming when $\mathrm{SNR} = 10 \text{ dB}$. We observe that the approximation error between the full-digital and hybrid beamformers, which is $\sum_{k = 1}^{K} \Vert \widetilde{\mathbf{F}}_{k}^{*} - \mathbf{F}_{\mathrm{RF}} \mathbf{F}_{\mathrm{BB}k} \Vert_{\mathrm{F}}^{2}$, converges monotonically and well in a few iterations.

\figurename{\ref{Fig3a}} presents the average SE of each multi-antenna user in the considered system versus SNR. We observe that the proposed design achieves the SE close to the Mink-L1 method and better than the other hybrid beamforming methods for DSs and FSs. Given that the Mink-L1 method cannot guarantee that each RF chain is connected to at least one antenna, our proposed design has its superiority. 
The SE of our proposed design is inferior to the HBF-FCS method mainly due to the better beamforming gains in the FCSs. When $\mathrm{SNR} = 20 \text{ dB}$, the proposed design can achieve approximately $79\%$ SE of the HBF-FCS method. 
Our proposed design is generalized and can also be used in the scenario with single-antenna users. In \figurename{\ref{Fig3b}}, we evaluate the average SE versus SNR in the case where $N_{\mathrm{R}} = 1, N_{\mathrm{s}} = 1$, and $K = 16$. It is seen that the proposed design outperforms the DC-HP (ZF) method in \cite{Dyanmic_1Ant}. 

\figurename{\ref{Fig3c}} illustrates the sum SE of all users versus the number of users when $\mathrm{SNR} = 10 \text{ dB}$. It is observed that with the increase in users, the sum SE achieved by the proposed design first increases and then decreases, which reaches the maximum at the number of users $K = 6$. This is because when $K$ grows beyond $6$, the IUI increases severely so that the sum SE degrades gradually. The sum SE of the proposed design is better than the DHBF, LowCom, and SDR-AltMin methods.  
However, when $K$ is large, the Mink-L1 method achieves a little better SE than our proposed design. The probable reason is that the Mink-L1 method cannot guarantee that each RF chain is connected to at least one antenna, making this method less vulnerable to the IUI than our proposed design.

\vspace{-0.4cm}
\section{Conclusion} 
\vspace{-0.15cm}
In this letter, we proposed an effective three-stage hybrid beamforming design for mmWave MU-MIMO systems with DSs, which can guarantee that each RF chain is connected to at least one antenna and effectively null out the IUI. Numerical results demonstrated that the proposed design can achieve better performance than the related methods for DSs and FSs.

\ifCLASSOPTIONcaptionsoff
  \newpage
\fi
\vspace{-0.45cm}
\normalem
\bibliographystyle{IEEEtran}
\bibliography{Reference}
\end{document}